\documentclass{article}
\usepackage{amsmath}
\usepackage{amssymb}
\usepackage{amsthm}
\usepackage{mathtools}
\usepackage{hyperref}
\usepackage{cleveref}
\usepackage{authblk}
\usepackage{tikz}
\usetikzlibrary{arrows}
\usepackage{comment}
\usepackage{enumerate}
\usepackage{enumitem}

\def\OPT{\textsf{OPT}}

\def \cost{\textsf{cost}}
\def \effect{\textsf{effect}}

\newtheorem{lemma}{Lemma}
\newtheorem{theorem}{Theorem}

\newtheorem{corollary}{Corollary}

\newtheorem{remark}{Remark}

\newtheorem{example}{Example}

\author{Robert Chiang}

\affil{Department of Combinatorics and Optimization, University of Waterloo,\\ 200 University Avenue West,
Waterloo, ON, Canada N2L 3G1\\
  \texttt{rjchiang@uwaterloo.ca}}
  \author{Kanstantsin Pashkovich}
 \affil{School of Computer Science and Electrical Engineering, University of Ottawa,\\ 800 King Edward Avenue,
Ottawa, ON, Canada
K1N 6N5\\
    \texttt{kpashkov@uottawa.ca}}

\title{On the approximability of the stable matching problem with ties of size two}

\date{}

\begin{document}
\maketitle

\begin{abstract}
The stable matching problem is one of the central problems of algorithmic game theory. If participants are allowed to have ties, the problem of finding a stable matching of maximum cardinality is an $\mathcal{NP}$-hard problem, even when the ties are of size two. Moreover, in this setting it is UGC-hard to provide an approximation with a constant factor smaller than $4/3$. In this paper,  we give a tight analysis of an approximation algorithm given by Huang and Kavitha for the maximum cardinality stable matching problem with ties of size two, demonstrating an improved $4/3$-approximation factor. 
\end{abstract}

\section{Introduction}

The stable matching problem is of crucial importance for the game theory.  In an instance of a maximum cardinality stable matching problem we are given a bipartite graph $G=(A\cup B, E)$ with bipartition $A$ and $B$. Following the standard terminology, we refer to $A$ as men and $B$ as women. For $a\in A$, we define $N(a)$ to be the subset of nodes in $B$ adjacent to $a$ in $G$; analogously we define $N(b)$ for $b\in B$.

Each person $c\in A\cup B$ has strict preferences over $N(c)$. A matching $M$ is called a \emph{stable matching} if there are no $a\in A$ and $b\in B$ such that $a$ is either unmatched or prefers $b$ to the woman he is matched to by $M$ and $b$ is either unmatched or prefers $a$ to the man she is matched to by $M$. Clearly, if a matching $M$ is not stable, then it contains a pair $(a,b)$, $a\in A$, $b\in B$ satisfying the above conditions; such a pair $(a,b)$ is called a \emph{blocking pair} for $M$.

In their seminal work, Gale and Shapley developed a polynomial running time algorithm to find a stable matching~\cite{GaleShapley62}. Moreover, since all stable matchings have the same cardinality~\cite{GaleSotomayor85}, the algorithm of Gale and Shapley finds a maximum cardinality stable matching in polynomial running time. The situation changes when the people are allowed to have \emph{ties}. In the case of ties, stable matchings for the same instance of a problem can have different cardinalities. Moreover, it is $\mathcal{NP}$-hard to find a maximum cardinality stable matching even when there are ties of size two only~\cite{Manlove}. In this case, it is $\mathcal{NP}$-hard  to approximate the maximum cardinality of a stable matching with a constant factor smaller than $21/19$ and UGC (Unique Game Conjecture)-hard to approximate with a constant factor smaller than $4/3$~\cite{Yanagisawa}. We would like to say that the maximum cardinality stable matching problem with ties appears in diverse situations and thus approximation algorithms for different variants of this problem were extensively studied 
\cite{Bauckholt, HuangKavitha2015, Iwamaetal07, Iwamaetal14,  Kiraly11,  algorithms, LP18,  mcdermid_first_approx, paluch_original}.

To the best of our knowledge, the algorithm with a currently best approximation factor for the maximum cardinality stable matching problem with ties of size two is due to Huang and Kavitha~\cite{HuangKavitha2015}. In their paper, Huang and Kavitha provided an approximation algorithm for this problem and showed that the approximation factor of their algorithm is at most $10/7$.

\subsection*{Our Contribution} In this paper,  we give a tight analysis of the approximation algorithm given by Huang and Kavitha~\cite{HuangKavitha2015} for the maximum cardinality stable matching problem with ties of size two, demonstrating an improved $4/3$-approximation bound. Notably, any polynomial running time algorithm with a smaller approximation factor than $4/3$ would automatically lead to refutation of the Unique Game Conjecture~\cite{Yanagisawa}.

To obtain our result we use a new charging scheme. In contrast to the charging scheme in~\cite{HuangKavitha2015}, our charging scheme is ``local". In particular, the charging scheme in~\cite{HuangKavitha2015} creates charges from paths and distributes the charges along paths (using so called ``good paths"), distributing the charges ``globally". The charging scheme in~\cite{HuangKavitha2015} for the algorithm for the problem with \emph{ties of size two} closely follows the original charging scheme in~\cite{HuangKavitha2015} for the algorithm for the problem with  \emph{one-sided ties}. Recently, the charging scheme for the problem with one-sided ties was substantially modified by Bauckholt,  Pashkovich and  Sanit\`{a}  in~\cite{Bauckholt}, where the modified charging scheme still has a ``global" nature and distributes the charges along paths but in a more nuanced way (using so called ``path jumps",  ``matching jumps" and ``matching jumps with exception") than the original charging scheme in~\cite{HuangKavitha2015}. However, it is not clear whether the original charging scheme in~\cite{HuangKavitha2015} or new ideas coming from the modified charging scheme in~\cite{Bauckholt} could lead to a better analysis of the algorithm for the problem with ties of size two. Our charging scheme is different from the two charging schemes above, in particular it is local and therefore it is much simpler to work with. To implement our charging scheme, we need to know only the local structure of the output matching, local structure of an optimal matching, and local structure of accepted proposals at the end of the algorithm. \emph{To sum up, in our charging scheme nodes get charges only from the proposals involving them, while in the charging schemes in~\cite{Bauckholt} and \cite{HuangKavitha2015} the charges come from paths and are distributed to nodes along paths.} A detailed comparison of our approach with approaches from~\cite{Bauckholt} and~\cite{HuangKavitha2015} can be found in Appendix.

\section{Algorithm by Huang and Kavitha}
\label{section:algorithm}

First, let us describe the algorithm by Huang and Kavitha~\cite{HuangKavitha2015}. The description of the algorithm closely follows the description of the algorithm by Huang and Kavitha~\cite{HuangKavitha2015} for the maximum cardinality stable matching problem with one-sided ties from~\cite{Bauckholt}.

A stable matching is computed in two phases. In the first phase, which is the proposal
phase, men (in arbitrary order) make proposals to women, 
while women accept, bounce, forward, or reject proposals. In the second phase,  we consider a graph based on the proposals at the end of the first phase, which in its turn leads to an output stable matching.

Before we describe the algorithm by Huang and Kavitha~\cite{HuangKavitha2015}, let us introduce the following notions related to preferences. Let $a\in A$ and $a'\in A$ be on the preference list of $b\in B$; i.e., $a$ and $a'$ are in $N(b)$. Then we can compare $a$ and $a'$ from the perspective of $b$. If \(a\) and \(a'\) are tied on the list of \(b\), we say \(b\) is \emph{indifferent} between them, denoted by \(a \simeq_b a'\). Further, if $b$ ranks $a$ strictly higher than $a'$ on her preference list, then we say that $b$ \emph{(strongly) prefers} $a$ to $a'$, denoted by $a >_b a'$; otherwise, we say that $b$ \emph{weakly prefers} $a'$ to $a$, denoted by $a' \geq_b a$. In other words, if $b$ weakly prefers $a'$ to $a$ then $b$ is either indifferent between $a$ and $a'$ or strongly prefers $a'$ to $a$. Analogously, we define \emph{indifference} $b'\simeq_a b$, \emph{weak preference} $b'\geq_a b$, and \emph{(strong) preference} $b'>_a b$ for men over women.

\subsubsection*{Proposals} Each man $a\in A$ has two proposals $p^1_a$ and $p^2_a$. Initially, both $p^1_a$ and $p^2_a$ are offered to the first woman on $a$'s list. At each moment of the algorithm every man has one of the following three \emph{statuses}: \emph{basic}, \emph{$1$-promoted}, and \emph{$2$-promoted}.  Each man keeps a rejection history to record the women who have rejected him in his current status.

If a proposal $p^i_a$, $i=1,2$ is rejected by a woman $b$ on $a$'s list, then the proposal $p^i_a$ goes to a most preferred  woman on $a$'s list who has not rejected $a$ in his current status. In the case when there is no woman on $a$'s list who has not rejected $a$ in his current status,  the man $a$ changes his status as described below or stops making proposals. If the man $a$ changes his status, his rejection history is emptied and  $a$ starts making proposals again by proposing to a most preferred woman on his list. Note that in the case where there are two most preferred women on $a$'s list who have not rejected him in his current status, the man $a$ breaks the tie arbitrarily.

Each man $a\in A$ starts as a basic man. If every woman in $N(a)$ rejects a proposal of $a$ at least once, the man $a$ becomes $1$-promoted. If afterwards every woman in $N(a)$ again rejects a proposal of $a$ as a $1$-promoted man at least once, the man $a$ becomes $2$-promoted. Finally, if every woman in $N(a)$ rejects a proposal of $a$ as a $2$-promoted man at least once, the man $a$ stops making proposals.

\subsubsection*{Proposals' Acceptance}
 A woman $b$, who gets a proposal $p^i_a$ from a man $a$, always accepts it if at the moment $b$ holds at most one proposal excluding $p^i_a$. Otherwise, $b$ tries to make a \emph{bounce} step, and if the bounce step is not successful, $b$ tries to make a \emph{forward} step.

\begin{itemize}
\item \emph{Bounce step:} So the woman $b$ at the moment holds two proposals, $p^{i'}_{a'}$ and $p^{i''}_{a''}$, and receives a third proposal $p^{i}_{a}$. If for some $\alpha\in \{a,a',a''\}$ there exists a woman $\beta$ such that $b\simeq_{\alpha}\beta$ and at the moment $\beta$ holds at most one proposal, then a proposal from $\alpha$ to $b$ is bounced to $\beta$ and the bounce step is called \emph{successful}. 
 
\item \emph{Forward step:} So the woman $b$ at the moment holds two proposals, $p^{i'}_{a'}$ $p^{i''}_{a''}$, and receives a proposal $p^{i}_{a}$. If two of the proposals in $\{p^i_a,\,p^{i'}_{a'},\, p^{i''}_{a''}\}$ are from the same man $\alpha$ and there exists a woman $\beta$ distinct from $b$ such that  $b\simeq_\alpha \beta$ and $\beta$ has not rejected $\alpha$ in his current status, then a proposal $p^1_\alpha$ from $\alpha$ to $b$ is forwarded to $\beta$ and  the forward step is called \emph{successful}. 
\end{itemize}
 
If $b$ holds proposals $p^{i'}_{a'}$ and $p^{i''}_{a''}$ and receives a different proposal $p^i_a$, but both the bounce and forward steps are not successful, then $b$ rejects any of the \emph{least desirable proposals} as defined below, breaking ties arbitrarily.{ \it Note that a proposal of a man which is bounced or forwarded is not considered as rejected, so there is no update of the rejection history for any of the men during the bounce or forward steps.}

For a woman $b$, proposal $p^i_a$ is superior to $p^{i'}_{a'}$ if one of the following is true
\begin{itemize}
\item $b$ prefers $a$ to $a'$.
\item $b$ is indifferent between $a$ and $a'$; $a$ is currently $2$-promoted while $a'$ is not $2$-promoted.
\item $b$ is indifferent between $a$ and $a'$; $a$ is currently $1$-promoted while $a'$ is basic.
\end{itemize}
A proposal $p^i_a$ is a \emph{least desirable proposal} among a set of proposals that a woman $b$ has if it is not superior to any of the proposals from which the woman $b$ selects two proposals to keep.

There is also a special case for the rejection step. Assume that $b$ holds two proposals $p^{i'}_{a'}$, $p^{i''}_{a''}$ and receives a different proposal $p^i_a$, where $a\simeq_b a'\simeq_b a''$ and the promotion statuses of $a$, $a'$, and $a''$ are the same. Clearly, since the ties are of size $2$, we have that two of the three proposals  $p^i_a$, $p^{i'}_{a'}$, and $p^{i''}_{a''}$ are from the same man. In this special case, $b$ rejects one of the proposals from this man. 
 
\subsubsection*{The output matching}

Let $G'$ be the bipartite graph with the node set $A\cup B$ and the edge set $E'$, where $E'$ consists of the edges $(a,b)$, $a\in A$, and $b\in B$ such that at the end of the algorithm $b$ holds a proposal from $a$. Note for the sake of exposition that we allow $G'$ to contain two parallel edges of the form $(a,b)$, $a\in A$ and $b\in B$, when at the end of the algorithm $b$ holds two proposals from $a$. Clearly, the degree of a node in $G'$ is at most two, since each man has at most two proposals and each woman is holding at most two proposals at any point in time. Let $M$ be a maximum cardinality matching in $G'$ where all degree two nodes of $G'$ are matched. In~\cite{HuangKavitha2015}, it was shown that the matching $M$ is a stable matching in the graph $G$.

\begin{theorem}[\cite{HuangKavitha2015}]
The total number of proposals made during the algorithm is~$O(|E|)$  and the output matching is a stable matching in $G=(A\cup B, E)$.  
\end{theorem}

\section{Tight Analysis}
\label{sec:tight_analysis}

Let $\OPT$ be a stable matching with the maximum cardinality, and let \(M\) be a stable matching output by the algorithm. If $b\in B$ is matched with $a\in A$ in $\OPT$, we use the following notation: $\OPT(b):= a$ and $\OPT(a):= b$. Similarly, if $b\in B$ is matched with $a\in A$ in $M$, we use the notation $M(b):= a$ and $M(a):=b$. For the sake of exposition we also define $M(a):=\varnothing$, $M(b):=\varnothing$, $\OPT(a):=\varnothing$, and $\OPT(b):=\varnothing$ when $a$ is not matched by $M$, $b$ is not matched by $M$, $a$ is not matched by $\OPT$, and $b$ is not matched by $\OPT$, respectively. Recall that for every $a\in A$ and $b\in N(a)$ we have $b>_a\varnothing$. Similarly, for every $b\in B$ and $a\in N(b)$ we have $a>_b\varnothing$.

A woman $b\in B$ is  called \emph{successful} if she holds two proposals at the end of the algorithm, i.e., $b$ is successful if the degree of $b$ in $G'$ is two. Similarly, a man $a\in A$ is \emph{successful} if both of his proposals are accepted at the end of the algorithm, i.e., $a$ is successful if the degree of $a$ in $G'$ is two. We call a person \emph{unsuccessful} if that person is not \emph{successful}. Further, if during the algorithm a woman rejected a proposal, we call her \emph{popular}; otherwise, we call her \emph{unpopular}. 

The next remark follows directly from the algorithm of Huang and Kavitha~\cite{HuangKavitha2015} and helps to understand the bouncing step better.
 \begin{remark}\label{rem:bouncing_step}
 Let $a\in A$ and $b, b'\in B$ be such that at the end of the algorithm $b$ holds a proposal from $a$, $b'$ is unsuccessful, and $b'\simeq_a b$. Then $b$ is unpopular.
 \end{remark}
 \begin{proof}
If at some point during the algorithm $b$ rejected a proposal, then at that point $b$ had no proposal for a successful bouncing step. Thus, at any later time point if $b$ received a new proposal that could be successfully bounced, then this new proposal would have been bounced. This contradicts the fact that at the end of the algorithm $b$ has a proposal from $a$, which can be successfully bounced.
\end{proof}

\subsection{Inputs and Outputs}

Inputs and outputs are central objects for our charging scheme, i.e. for defining a cost of a node in $G$.

Given a woman $b\in B$, we say that an edge $e$ incident to $b$ in $G'$ is an \emph{input to $b$} if  $e$ is not in $\OPT$ and not in $M$. Given a man $a\in A$, an \emph{output from $a$} is an edge $e$ incident to $a$ in $G'$ such that $e$ is not in $\OPT$ and not in $M$. In other words, an edge $(a,b)$ in $G'$ is an input to $b\in B$ and an output from $a\in A$ if $(a,b)$ is not in $\OPT+M$, otherwise the edge $(a,b)$ is neither input nor output.

Note that if in $G'$ there are two parallel edges $e_1$ and $e_2$ of the form $(a,b)$, then $M$ must contain the edge $(a,b)$. In this case, one of the parallel edges $\{e_1, e_2\}$ is associated with $M$, while the other is counted as an input and an output unless $\OPT$ also contains $(a,b)$, in which case it is associated with $\OPT$. 

An input $(a,b)$ to $b\in B$ is called a \emph{bad input} if one of the following is true:
\begin{itemize}
	\item $b$ is popular and \(a >_b \OPT(b)\).
	\item $b$ is popular and \(a \simeq_b \OPT(b)\), but $\OPT(b)$ is unsuccessful.
\end{itemize}
An input $(a,b)$ to $b\in B$ is a \emph{good input} to \(b\) if it is not a bad input. In other words, an input $(a,b)$ to $b\in B$ is a \emph{good input} if one of the following is true:
\begin{itemize}
	\item \(b\) is unpopular.
	\item \(b\) is popular and \(\OPT(b) >_b a\).
	\item \(b\) is popular and \(\OPT(b) \simeq_b a\), and \(\OPT(b)\) is successful.
\end{itemize}
An output $(a,b)$ from a man $a$ is called a \emph{bad output}  if one of the following is true:
\begin{itemize}
	\item \(b >_a \OPT(a)\) and \(a\) is not $2$-promoted.
	\item \(b \simeq_a \OPT(a)\), but \(\OPT(a)\) is unsuccessful. 
\end{itemize}
A \emph{good output} from a man \(a\) is an output \((a,b)\) that is not a bad output. In other words, an output $(a,b)$ from $a\in A$ is a \emph{good output} if one of the following is true:
\begin{itemize}
	\item \(b >_a \OPT(a)\) and \(a\) is $2$-promoted.
	\item \(b \simeq_a \OPT(a)\) and \(\OPT(a)\) is successful.
	\item \(b <_a \OPT(a)\).
\end{itemize}

\begin{lemma}\label{lem:double_bad}
There is no edge which is both a bad input and a bad output.
\end{lemma}

\begin{proof}
Assume that an edge \((a,b)\), $a\in A$, $b\in B$ is both a bad input to \(b\) and a bad output from \(a\). From the definitions of bad inputs and bad outputs, we have that one of the following is true
\begin{enumerate}[label*=\roman*.]
	\item \label{case1_lem:double_bad} \(a >_b \OPT(b)\); \(b >_a \OPT(a)\).
	\item \label{case2_lem:double_bad} \(a \simeq_b \OPT(b)\) and \(\OPT(b)\) is unsuccessful; \(b >_a \OPT(a)\) and \(a\) is not $2$-promoted. 
	\item \label{case3_lem:double_bad} \(b \simeq_a \OPT(a)\) and \(\OPT(a)\) is unsuccessful; \(b\) is popular. 
\end{enumerate}

In Case~\eqref{case1_lem:double_bad}, the edge $(a,b)$ is a blocking pair for $\OPT$, contradicting the stability of $\OPT$.

In Case~\eqref{case2_lem:double_bad}, since \(\OPT(b)\) is unsuccessful, \(\OPT(b)\) is 2-promoted and was rejected by $b$ in that status. On the other hand, \(a \simeq_b \OPT(b)\),  $a$ is not $2$-promoted, and at the end of the algorithm $b$ holds a proposal from $a$,  contradicting the fact that $b$ rejected a proposal from  $\OPT(b)$ when he was $2$-promoted. 

In Case~\eqref{case3_lem:double_bad}, we have a contradiction to Remark~\ref{rem:bouncing_step} where $b':=\OPT(a)$.

\end{proof}

\begin{corollary}\label{cor:total_effectiveness}
The number of good inputs is not smaller than the number of bad outputs.
\end{corollary}

\begin{proof}
Let us assume the contrary, i.e., let us assume that the number of bad outputs is larger than the number of good inputs. Then there is an edge in $G'$ which is a bad input and is not a good output. In other words, there is an edge which is both a bad input and a bad output, contradicting Lemma~\ref{lem:double_bad}.
\end{proof}

\subsection{Cost, Effectiveness}

For a man $a\in A$, we define his \emph{cost} as follows:
\[
\cost(a):=\begin{cases}
				\deg(a)+1&\text{if}\quad a \quad\text{has a bad output}\\
				\deg(a)&\text{otherwise}
		\end{cases}\,,
\]
For a woman $b\in B$, we define her \emph{cost} as follows
\[
\cost(b):=\begin{cases}
				\deg(b)-1&\text{if}\quad b \quad\text{has a good input}\\
				\deg(b)& \text{otherwise}
		\end{cases}\,,
\]
where $\deg$ denotes the degree of the corresponding node in $G'$. Note that if $G'$ has two parallel edges $(a,b)$, then we have $\deg(a)=\deg(b)=2$. For a node set $S\subseteq A\cup B$, we define the \emph{cost} of $S$ as the sum of costs of all the nodes in $S$, i.e.,
\[
\cost(S):=\sum_{v\in S}\cost(v)\,,
\]
and the \emph{effectiveness} of $S$ as the number of good inputs to the nodes in $S$ minus the number of bad outputs from the nodes in $S$, i.e.,
\[
\effect(S):= \sum_{v\in S}\deg(v)-\cost(S)\,.
\]

\begin{lemma}\label{lem:3cost}
Let \(a\in A\) and \(b\in B\) be such that \(\OPT(a) = b\). Then \(\cost(\{a,b\}) \geq 2\) holds. Moreover, if $\deg(a)\geq 1$, then \(\cost(\{a,b\}) \geq 3\) holds. 
\end{lemma}

\begin{proof}
One of the following is true:
 \begin{enumerate}[label*=\roman*.]
 	\item \label{case0_lem:3cost} $\deg(a)=0$.
	\item \label{case1_lem:3cost} $\deg(a)\leq 1$ and $\deg(b)\leq 1$.
	\item \label{case2_lem:3cost} $\deg(a) = 2$ and $\deg(b)=0$.
	\item \label{case3_lem:3cost} $\deg(a) = 2$ and $\deg(b)\geq 1$.
	\item \label{case4_lem:3cost}  $\deg(a) = 1$, $\deg(b) = 2$ and $b$ has a bad input or no input.
	\item \label{case5_lem:3cost}  $\deg(a) = 1$, $\deg(b) = 2$ and $b$ has a good input.
\end{enumerate}

In Case~\eqref{case0_lem:3cost}, $a$ is unsuccessful. Since $(a,b)$ is an edge in $G$, $b$ rejected a proposal from $a$ at least once and so $b$ is popular. Thus, there are two distinct edges $(a',b)$, $(a'',b)$ in $G'$ such that $a\leq_b a'$ and $a\leq_b a''$. It is straightforward to check that neither $(a',b)$ nor $(a'',b)$ is a good input. Thus we have $\cost(b)= 2$, implying the desired inequality of $\cost(\{a,b\})\geq \cost(b)=2$. 

In Case~\eqref{case1_lem:3cost},  $a$ is unsuccessful. Since $(a,b)$ is an edge in $G$, $b$ rejected a proposal from $a$ at least once and so $b$ is popular. However, $\deg(b)\leq1$ and so $b$ is not popular, contradiction.

In Case~\eqref{case2_lem:3cost}, $a$ did not ever propose to $b$ during the algorithm. Thus in $G'$ there are two distinct edges $(a,b')$ and $(a,b'')$ with $b\leq_a b'$, $b\leq_a b''$. Since $\deg(b)=0$, one of the edges $(a,b')$, $(a,b'')$ is neither in $M$ nor in $\OPT$, giving rise to a bad output from $a$. Thus, we have $\cost(a)= 3$, implying $\cost(\{a,b\})=3\geq 3$. 

In Case~\eqref{case3_lem:3cost}, we have
\[
\cost(\{a,b\})=\underbrace{\cost(a)}_{\geq \deg(a)=2}+ \underbrace{\cost(b)}_{\geq 1} \geq  2+1= 3\,, 
\]
providing the desired inequality.

In Case~\eqref{case4_lem:3cost}, we have
\[
\cost(\{a,b\})=\underbrace{\cost(a)}_{\geq \deg(a)=1}+ \underbrace{\cost(b)}_{= \deg(b)=2} \geq  1+2=3\,, 
\]
providing the desired inequality.

In Case~\eqref{case5_lem:3cost}, $a$ is unsuccessful. Since $(a,b)$ is an edge in $G$, $b$ rejected a proposal from $a$ at least once and so $b$ is popular. Let $(a',b)$ be the good input to $b$. Since $b$ rejected a proposal from $a$, we have that $a'\geq_b a$.  Thus $b$ is popular, $\OPT(b)=a$ is unsuccessful, and $a'\geq_b a$, showing that the input $(a',b)$ is a bad input, contradiction.
\end{proof}

\subsection{Connected components in $\OPT+M$}

In this section, we start to relate the ratio between $|\OPT|$ and $|M|$ to the ratios between $|\OPT\cap C|$ and $\cost(C)$ for connected components in $\OPT+M$. For simplicity of exposition,  isolated nodes in $\OPT+M$ are considered by us as connected components in $\OPT+M$.

\begin{lemma}\label{lem:low_cost}
If \(|\OPT| / |M| > 4/3\), then there exists a connected component \(C\) in \(\OPT + M\) with \(\cost (C) < 3 |\OPT \cap C|\).
\end{lemma}

\begin{proof}
Assume for contradiction that \(|\OPT| / |M| > 4/3\) and that \(\cost (C) \geq 3 |\OPT \cap C|\) for all components \(C\). Then, we have
\[
\sum_{C \in \OPT + M} \cost (C) \geq 3|\OPT|\,,
\]
where the summation goes over the connected components in  $\OPT + M$.

On the other side, $M$ is a maximum cardinality matching in the graph $G'=(A\cup B, E')$, where each node in $G'$ has degree at most $2$. Hence, 
\[
|M|\geq |E'|/2=\sum_{v\in A\cup B}\deg(v)/4\,.
\]
Moreover, we have 
\[
\sum_{C \in \OPT + M}  \effect(C)=\effect(A\cup B)\geq 0\,,
\]
where the inequality is due to Corollary~\ref{cor:total_effectiveness}.
Combining the inequalities above, we obtain
\begin{align*}
4|M|\geq \sum_{v\in A\cup B}\deg(v)-\effect(A\cup B)=\cost(A\cup B)=\sum_{C \in \OPT + M} \cost (C) \\\geq 3|\OPT|\,,
\end{align*}
contradiction.
\end{proof}

\begin{remark}\label{rem:trivial_component}
Let  \(C\) be a connected component of $M+\OPT$ which consists of a single node. Then we have \(\cost (C) \geq 3 |\OPT \cap C|\).
\end{remark}

\subsubsection{Alternating Paths, Alternating Cycles and $\OPT$-Augmenting Paths}\label{sec:alternating_opt_augmenting}

The following corollary of Lemma~\ref{lem:3cost} establishes an upper bound on the number of $\OPT$ edges in a connected component of $M+\OPT$ in terms of the component's cost, applying only to alternating paths, alternating cycles, and $\OPT$-augmenting paths. Note that an edge which is in both $M$ and $\OPT$ is considered by us to be a trivial alternating cycle if it corresponds to two parallel edges in~$G'$.

\begin{corollary}\label{cor:nonaug}
Let \(C\) be a connected component of $M+\OPT$ that is an alternating path, alternating cycle, or $\OPT$-augmenting path. Then we have \(\cost (C) \geq 3 |\OPT \cap C|\). 
\end{corollary}

\begin{proof}
Let  \(C\) be a connected component of $M+\OPT$ which is an alternating path, alternating cycle, or $\OPT$-augmenting path. Then one of the following is true:
\begin{enumerate}[label*=\roman*.]
	\item \label{case1_cor:nonaug} $C$ is an $\OPT$-augmenting path of the form $a_0-b_1-\ldots-a_k-b_{k+1}$.
	\item \label{case2_cor:nonaug} $C$ is an alternating cycle of the form $a_1-b_1-\ldots-a_k-b_k-a_1$, where $(a_1,b_1)\in \OPT$.
	\item \label{case3_cor:nonaug} $C$ is an alternating path of the form $b_1-a_1-\ldots-b_k-a_{k}-b_{k+1}$, where $a_1\in A$ and $(a_1,b_1)\in \OPT$.
	\item \label{case4_cor:nonaug} $C$ is an alternating path of the form $a_1-b_1-\ldots-a_k-b_{k}-a_{k+1}$, where $a_1\in A$ and $(a_1,b_1)\in \OPT$. 
\end{enumerate}

In Cases~\eqref{case1_cor:nonaug}, \eqref{case2_cor:nonaug} and \eqref{case3_cor:nonaug}, by Lemma~\ref{lem:3cost}, for each $i=1,\ldots,k$ we have $\cost(\{a_i,b_i\})\geq 3$. Thus, we have the desired inequality:
\[
\cost (C)\geq\sum_{i=1}^k \underbrace{\cost(\{a_i,b_i\})}_{\geq 3} \geq 3 k=3 |\OPT \cap C|\,.
\]

In Case~\eqref{case4_cor:nonaug}, by Lemma~\ref{lem:3cost}, for each $i=2,\ldots,k$ we have $\cost(\{a_i,b_i\})\geq 3$ and also $\cost(\{a_1,b_1\})\geq 2$. Since $\cost(a_{k+1})\geq 1$, we have the desired inequality:
\[
\cost (C)=\underbrace{\cost(a_{k+1})}_{\geq 1}+\underbrace{\cost(\{a_1,b_1\})}_{\geq 2}+\sum_{i=2}^k \underbrace{\cost(\{a_i,b_i\})}_{\geq 3}\geq 3 k=3 |\OPT \cap C|\,.
\]
\end{proof}

The next remark is an immediate consequence of Lemma~\ref{lem:3cost}.

\begin{remark}
An edge $(a,b)$ that exists in both $M$ and $\OPT$ defines a connected component $C$ in $M+\OPT$. If $(a,b)$ does not correspond to a parallel edge in $G'$, then $\cost(C)\geq 3 =3| C\cap \OPT|$.
\end{remark}

\subsubsection{$M$-Augmenting Path}
In this section, we study connected components in $M+\OPT$ which are $M$-augmenting paths. First, we state the result of Huang and Kavitha~\cite{HuangKavitha2015}, showing that there exists no $M$-augmenting path of length $1$ or of length $3$ in $M+\OPT$. 
\begin{lemma}[\cite{HuangKavitha2015}]
\label{lemma:three_path_original_paper}
There is no $M$-augmenting path in $M+ \OPT$ of length $1$ or of length $3$.
\end{lemma}

Now let us consider $M$-augmenting paths in $M+\OPT$ of length at least~$5$. Given an $M$-augmenting path of the form $a_0-b_0-a_1-\ldots-a_k-b_k$ with $a_0\in A$, for $i=0,\ldots,k-1$ we say that \emph{$b_i$ points right} if one of the following is true:
\begin{itemize}
	\item \(a_{i+1}>_{b_i} a_i\).
	\item \(a_{i+1}\simeq_{b_i} a_i\) and at the end of the algorithm \(a_{i+1}\) is not basic.
\end{itemize}

\begin{lemma}\label{lem:interior}
Let $a_0-b_0-a_1-\ldots-a_k-b_k$ be an $M$-augmenting path in $M+ \OPT$ of length $2k+1$, $k\geq 2$, where $a_0\in A$. Then for every $i=1,\ldots, k-1$, at least one of the following is true:
\begin{enumerate}
	\item \label{case1_lem:interior_state}  \(\cost(\{a_i, b_i\}) \geq 4\).
	\item \label{case2_lem:interior_state}  \(b_i\) rejected a proposal from \(a_i\) during the algorithm and $b_i$ points right. 
	\item \label{case3_lem:interior_state}  \(b_{i-1}\) did not reject any proposals from \(a_{i-1}\) during the algorithm.
	\item \label{case4_lem:interior_state}  at the end of the algorithm, \(a_i\) is basic and \(b_{i-1} >_{a_i} b_i\). 
\end{enumerate}
\end{lemma}

\begin{proof}
Clearly, the edges \((a_{i}, b_{i-1})\) and \((a_{i+1}, b_i)\) exist in \(G'\) since these edges exist in \(M\). Thus we have \(\deg(a_i)\geq 1\) and \(\deg(b_i)\geq 1\). Moreover, since the edge \((a_i,b_i)\) exists in \(G\), at least one of the following is true: \(\deg(a_i)\geq 2\) or \(\deg(b_i)\geq 2\). Hence, it is enough to consider the following cases: 
\begin{enumerate}[label*=\roman*.]
	\item \label{case0_lem:interior} \(\deg(a_i) = 1\) and \(\deg(b_i)=2\).
  	\item \label{case1_lem:interior} \(\deg(a_i) = 2\).
  		\begin{enumerate}[label*=\roman*.]
		 	\item\label{case10_lem:interior} $b_i$ rejected a proposal from $a_i$ during the algorithm.
				\begin{enumerate}[label*=\roman*.]
   					\item\label{case101_lem:interior} $b_i$ has no good input.
					 \item\label{case102_lem:interior} $b_i$ has a good input.
  			\end{enumerate}
   			\item\label{case11_lem:interior} there is an edge \((a_i, b_i)\) in $G'$. 
			\item\label{case13_lem:interior} $a_i$ has a bad output.	
			\item\label{case14_lem:interior} $a_i$ has a good output; $b_i$ did not reject any proposals from $a_i$ during the algorithm.
  		\end{enumerate}
\end{enumerate}
We would like to note that while the above cases are not mutually exclusive, they cover all the possibilities.

In Case~\eqref{case0_lem:interior}, \(a_i\) is unsuccessful. Thus, $b_i$ rejected a proposal from $a_i$ as a $2$-promoted man. On the other hand, at the end of the algorithm $b_i$ has a proposal from $a_{i+1}$, implying \eqref{case2_lem:interior_state}.

In Case~\eqref{case101_lem:interior}, we have $\deg(b_i)=2$, since $b_i$ rejected a proposal from $a_i$ during the algorithm. Since $b_i$ has no good input, we have $\cost(b_i)=\deg(b_i)=2$. Thus, we have
\[
\cost(\{a_i,b_i\})=\underbrace{\cost(a_i)}_{\geq \deg(a_i)=2}+\underbrace{\cost(b_i)}_{=2}\geq 4\,,
\]
implying~\eqref{case1_lem:interior_state}.

In Case~\eqref{case102_lem:interior}, $b_i$ is popular, since $b_i$ rejected a proposal during the algorithm. Let $(a',b_i)$ be a good input to $b_i$. Since $(a',b_i)$ is a good input to $b_i$, we have $a_i\geq_{b_i}a'$ and $a'\neq a_i$.  Since during the algorithm $b_i$ rejected a proposal from $a_i$, but at the end of the algorithm has proposals from $a_{i+1}$ and $a'$, we have $a_i\leq_{b_i}a_{i+1}$ and $a_i\leq_{b_i}a'$. Because $a_i\geq_{b_i}a'$ and $a_i\leq_{b_i}a'$, we have $a_i\simeq_{b_i}a'$. 

Since $a_i\leq_{b_i}a_{i+1}$, $a_i\simeq_{b_i}a'$, $a'\neq a_i$ and ties are of size $2$, we either have $a_i<_{b_i}a_{i+1}$ or $a'= a_{i+1}$. If we have $a_i<_{b_i}a_{i+1}$, then \eqref{case2_lem:interior_state} holds. If $a'= a_{i+1}$ then at some point $b_i$ rejected a proposal of $a_i$ while at the end of the algorithm $b_i$ holds two proposals of $a_{i+1}$. If $a_{i+1}$ is basic, this contradicts the special case for the rejection step, because $a_i\simeq_{b_i}a_{i+1}$. Then $a_{i+1}$ is not basic, implying that $b_i$ points right and thus \eqref{case2_lem:interior_state}.
 
In Case~\eqref{case11_lem:interior}, \(\cost(a_i) = 2\) and \(\cost(b_i) = 2\), so \(\cost(\{a_i, b_i\}) =4\), implying~\eqref{case1_lem:interior_state}.

In Case~\eqref{case13_lem:interior}, we have $\cost(b_i)\geq 1$, since $b_i$ is matched by $M$ and hence $\deg(b_i) \geq 1$. We also have $\cost(a_i)=\deg(a_i)+1=3$, since $a_i$ has a bad output.  Thus, $\cost(\{a_i,b_i\})=\cost(a_i)+\cost(b_i)\geq 4$, showing \eqref{case1_lem:interior_state}.

In Case~\eqref{case14_lem:interior}, let $(a_i,b')$ be the good output from $a_i$. Since $b_i$ did not reject any proposals from $a_i$ during the algorithm, we have that $a_i$ is basic and also that $b_i\leq_{a_i}b_{i-1}$ and $b_i\leq_{a_i}b'$ hold. Since $a_i$ is basic and $(a_i,b')$ is a good output from $a_i$, we have that $b'\neq b_i$ and $b_i\geq_{a_i}b'$, and thus $b_i\simeq_{a_i}b'$. Since $b'\neq b_i$, $b_i\simeq_{a_i}b'$, $b_i\leq_{a_i}b_{i-1}$, and ties are of size two, we have $b_i<_{a_i}b_{i-1}$, implying~\eqref{case4_lem:interior_state}.

\end{proof}

\begin{corollary}\label{cor:induction}
Let $a_0-b_0-a_1-\ldots-a_k-b_k$ be an $M$-augmenting path in $M+ \OPT$ of length $2k+1$, $k\geq 2$, where $a_0\in A$. Then, for every $i=1,\ldots, k-1$  such that  \(\cost(\{a_i, b_i\}) = 3\), we have that if \(b_{i-1}\) rejected a proposal from \(a_{i-1}\) during the algorithm and \(b_{i-1}\) points right, then \(b_i\) rejected a proposal from \(a_i\) during the algorithm and \(b_i\) points right. 
\end{corollary}

\begin{proof}
By Lemma~\ref{lem:interior}, we have that at least one of the following statements is true:
\begin{enumerate}[label*=\roman*.]
	\item \label{case1_cor:induction}  \(\cost(\{a_i, b_i\}) \geq 4\).
	\item \label{case2_cor:induction}  \(b_i\) rejected a proposal from \(a_i\) during the algorithm and $b_i$ points right. 
	\item \label{case3_cor:induction}  \(b_{i-1}\) did not reject any proposal from \(a_{i-1}\) during the algorithm.
	\item \label{case4_cor:induction}  at the end of the algorithm \(a_i\) is basic and \(b_{i-1} >_{a_i} b_i\). 
\end{enumerate}

In Case~\eqref{case1_cor:induction}, we have a contradiction to \(\cost(\{a_i, b_i\}) =3\). In Case~\eqref{case3_cor:induction}, we have a contradiction to the fact that \(b_{i-1}\) rejected a proposal from \(a_{i-1}\) during the algorithm. Let us consider Case~\eqref{case4_cor:induction}. Since $a_i$ is basic and $b_{i-1}$ points right, we have \(a_i >_{b_{i-1}} a_{i-1}\). Thus, we have \(a_i >_{b_{i-1}} a_{i-1}\) and  \(b_{i-1} >_{a_i} b_i\), showing that  \((a_i, b_{i-1})\) is a blocking pair for \(\OPT\), contradicting the stability of $\OPT$.

In Case~\eqref{case2_cor:induction}, we obtain the desired statement immediately. 

\end{proof}

\begin{lemma}\label{lem:start}
Let $a_0-b_0-a_1-\ldots-a_k-b_k$ be an $M$-augmenting path in $M+ \OPT$ of length $2k+1$, $k\geq 2$, where $a_0\in A$. Then, we have \(\cost(\{a_0, b_0\}) \geq 2\). Moreover,  \(b_0\) rejected a proposal from \(a_0\) during the algorithm and \(b_0\) points right. 
\end{lemma}
\begin{proof}
Clearly, \(a_0\) is not matched by \(M\) and so \(a_0\) is unsuccessful. Thus, \(b_0\)  rejected a proposal from \(a_0\) as a $2$-promoted man. Hence, $b_0$ is popular and if $(a,b_0)$ is an edge in $G'$ then $a\geq_{b_0} a_0=\OPT(b_0)$. This demonstrates that \(b_0\) has no good input, and so \(\cost(\{a_0, b_0\}) \geq \cost(b_0) = 2\).

Now, at some point during the algorithm $b_0$ rejected a proposal from $a_0$ as a $2$-promoted man, but at the end of the algorithm $b_0$ has a proposal from $a_1$, showing that $b_0$ points right and finishing the proof.
\end{proof}

\begin{remark}\label{rem:lastfour}
Let $a_0-b_0-a_1-\ldots-a_k-b_k$ be an $M$-augmenting path in $M+ \OPT$ of length $2k+1$, $k\geq 2$, where $a_0\in A$. Then, we have \(\cost(\{a_k,b_k\}) \geq 3\).
\end{remark}

\begin{proof}
Clearly, $b_k$ is not matched by $M$ and so $b_k$ is unsuccessful. Hence, $a_k$ is basic and $a_k$ is successful. 

If at the end of the algorithm $b_k$ holds a proposal from $a_k$, then $b_k$ has no input. Thus, $\cost(b_k)=\deg(b_k)=1$ and $\cost(a_k)=\deg(a_k)=2$, showing that $\cost(\{a_k,b_k\})\geq 3$.

If at the end of the algorithm $b_k$ holds no proposal from $a_k$, then $a_k$  has an output $(a_k,b')$ with $b'\geq_{a_k} b_k$. Thus, $(a_k,b')$ is a bad output, implying that $\cost(a_k)=\deg(a_k)+1=3$.
\end{proof}

\begin{lemma}\label{lem:end}
Let $a_0-b_0-a_1-\ldots-a_k-b_k$ be an $M$-augmenting path in $M+ \OPT$ of length $2k+1$, $k\geq 2$, where $a_0\in A$.
If \(b_{k-2}\) rejected a proposal from  \(a_{k-2}\) during the algorithm and  \(b_{k-2}\)  points right, then \(\cost(\{a_{k-1}, b_{k-1}, a_k, b_k\}) \geq 7\).
\end{lemma}

\begin{proof}
For the proof of the lemma  it is enough to consider following cases
\begin{enumerate}[label*=\roman*.]
	\item \label{case1_lem:end} $b_{k-1}$ did not reject any proposals from $a_{k-1}$ during the algorithm.
	\item \label{case2_lem:end} $b_{k-1}$ rejected a proposal from $a_{k-1}$ during the algorithm.
		\begin{enumerate}[label*=\roman*.]
			\item \label{case21_lem:end} there are two parallel edges \((a_k, b_{k-1})\) in $G'$.
			\item \label{case22_lem:end} there is an edge \((a_{k-1}, b_{k-1})\) in $G'$.
			\item \label{case23_lem:end} there is an edge $(a',b_{k-1})$ in $G'$ such that $a'\neq a_{k-1}$ and $a'\neq a_k$.
		\end{enumerate}
\end{enumerate}

In Case~\eqref{case1_lem:end}, since $b_{k-1}$ did not reject any proposal from $a_{k-1}$ during the algorithm, we have that $a_{k-1}$ is basic and $M(a_{k-1})=b_{k-2}\geq_{a_{k-1}} b_{k-1}$. Since $a_{k-1}$ is basic and $b_{k-2}$ points right, we have $a_{k-1}>_{b_{k-2}} a_{k-2}$. 

Moreover, we have $b_{k-2}\leq_{a_{k-1}} b_{k-1}$, since otherwise $(a_{k-1},b_{k-2})$ is a blocking pair for $\OPT$, contradicting the stability of $\OPT$. Because $b_{k-2}\geq_{a_{k-1}} b_{k-1}$ and $b_{k-2}\leq_{a_{k-1}} b_{k-1}$ hold, we have $b_{k-2}\simeq_{a_{k-1}} b_{k-1}$.

The forwarding step excludes the possibility of $b_{k-2}$ having two proposals from $a_{k-1}$ at the end of the algorithm while also having rejected a proposal at some point during the algorithm. On the other hand, we have that $a_{k-1}$ is basic and by the statement of the lemma $b_{k-2}$ rejected a proposal from $a_{k-2}$ during the algorithm, showing that there exists an edge $(a_{k-1},b')$ in $G'$ such that $b'\neq b_{k-2}$. Since $b_{k-1}$ did not reject any proposal from $a_{k-1}$ during the algorithm, we have $b'\geq_{a_{k-1}} b_{k-1}$. 

Now, because $b'\geq_{a_{k-1}} b_{k-1}$, $b_{k-2}\simeq_{a_{k-1}} b_{k-1}$, $b'\neq b_{k-2}$, and ties are of size two, we have that either $b'=b_{k-1}$ or $b'>_{a_{k-1}} b_{k-1}$. If $b'=b_{k-1}$, we have $\cost(a_{k-1})=\deg(a_{k-1})=2$ and $\cost(b_{k-1})=\deg(b_{k-1})=2$, implying the desired inequality
\[
\cost(\{a_{k-1}, b_{k-1}, a_k, b_k\})\geq \underbrace{\cost(a_{k-1})}_{= 2}+\underbrace{\cost(b_{k-1})}_{=2}+\underbrace{\cost(\{a_k,b_k\})}_{\substack{\geq 3\\\text{by Remark~\ref{rem:lastfour}}}}\geq7\,.
\]
If $b'>_{a_{k-1}} b_{k-1}$, then $(a_{k-1},b')$ is a bad output from $a_{k-1}$, because $a_{k-1}$ is basic. Hence, $\cost(a_{k-1})=\deg(a_{k-1})+1=3$ and $\cost(b_{k-1})\geq 1$, since $b_{k-1}$ is matched by $M$. This also gives us the desired inequality
\[
\cost(\{a_{k-1}, b_{k-1}, a_k, b_k\})\geq \underbrace{\cost(a_{k-1})}_{= 3}+\underbrace{\cost(b_{k-1})}_{\geq 1}+\underbrace{\cost(\{a_k,b_k\})}_{\substack{\geq 3\\\text{by Remark~\ref{rem:lastfour}}}}\geq 7\,. 
\]

In Case~\eqref{case21_lem:end}, since at the end of the algorithm $b_{k-1}$ holds two proposals from $a_k$, who is basic, and at some point of the algorithm $b_{k-1}$ rejected a proposal from $a_{k-1}$, we have that $a_k\geq _{b_{k-1}} a_{k-1}$. The special case of the rejection step excludes the possibility that $a_k\simeq _{b_{k-1}} a_{k-1}$, implying that $a_k> _{b_{k-1}} a_{k-1}$.

However, $b_k$ is unsuccessful and there is an edge $(a_k,b_{k-1})$ in $G'$, hence $a_k$ is basic and $b_{k-1}\geq_{a_k} b_k$. If $b_{k-1}>_{a_k} b_k$, then $(a_k, b_{k-1})$ is a blocking pair for $\OPT$, contradicting the stability of $\OPT$. Thus, we have $b_{k-1}\simeq_{a_k} b_k$. But by Remark~\ref{rem:bouncing_step}, $b_{k-1}\simeq_{a_k} b_k$ and $\deg(b_k)=1$ together with the bouncing step excludes the possibility that $b_{k-1}$ rejects any proposal during the algorithm and also holds two proposals from $a_k$ at the end, contradiction.

In Case~\eqref{case22_lem:end}, we have $\cost(a_{k-1})=\deg(a_{k-1})=2$ and $\cost(b_{k-1})=\deg(b_{k-1})=2$, implying the desired inequality
\[
\cost(\{a_{k-1}, b_{k-1}, a_k, b_k\})\geq \underbrace{\cost(a_{k-1})}_{= 2}+\underbrace{\cost(b_{k-1})}_{=2}+\underbrace{\cost(\{a_k,b_k\})}_{\substack{\geq3\\\text{by Remark~\ref{rem:lastfour}}}}\geq7\,.
\]

In Case~\eqref{case23_lem:end}, since at the end of the algorithm $b_{k-1}$ has proposals from $a'$ and $a_{k}$ and at some point of the algorithm $b_{k-1}$  rejected a proposal from $a_{k-1}$, we have  $a_k \geq_{b_{k-1}} a_{k-1}$ and $a' \geq_{b_{k-1}} a_{k-1}$. Since $a'\neq a_{k-1}$ and $a'\neq a_k$ and ties are of size two, we have either $a_k >_{b_{k-1}} a_{k-1}$ or $a' >_{b_{k-1}} a_{k-1}$. 

If $a_k >_{b_{k-1}} a_{k-1}$, then we have $b_{k-1}\leq_{a_k} b_k$, since otherwise $(a_k, b_{k-1})$ is a blocking pair for $\OPT$, contradicting the stability of $\OPT$. On the other side, $b_k$ is unsuccessful and an edge $(a_k, b_{k-1})$ exists in $G'$, so we have $b_{k-1}\geq_{a_k} b_k$. Thus, we have $b_{k-1}\simeq_{a_k} b_k$. However, by Remark~\ref{rem:bouncing_step}, $b_{k-1}\simeq_{a_k} b_k$ and $\deg(b_k)=1$ together with the bouncing step exclude the possibility that $b_{k-1}$ rejects any proposal during the algorithm but at the end of the algorithm has a proposal from $a_k$, contradiction.

Let us now consider the case when $a' >_{b_{k-1}} a_{k-1}$ and $a_k \simeq_{b_{k-1}} a_{k-1}$. Then $(a', b_{k-1})$ is a bad input to $b_{k-1}$, because $a'\neq a_{k-1}$, $a'\neq a_k$, $a' >_{b_{k-1}} a_{k-1}$ and $b_{k-1}$ rejected a proposal during the algorithm. Hence, we have $\cost(b_{k-1})=\deg(b_{k-1})=2$. Also due to the fact that $a_k \simeq_{b_{k-1}} a_{k-1}$ and the fact that at the end of the algorithm $b_{k-1}$ has a proposal from $a_k$ as a basic man, we have that $a_{k-1}$ is successful. Since $a_{k-1}$ is successful, we have $\cost(a_{k-1})\geq \deg(a_{k-1})=2$, implying the desired inequality
\[
\cost(\{a_{k-1}, b_{k-1}, a_k, b_k\})\geq \underbrace{\cost(a_{k-1})}_{\geq 2}+\underbrace{\cost(b_{k-1})}_{=2}+\underbrace{\cost(\{a_k,b_k\})}_{\substack{\geq3\\\text{by Remark~\ref{rem:lastfour}}}}\geq7\,.
\]
\end{proof}

\begin{lemma}\label{lem:aug}
Let \(C\) be a connected component of $M+\OPT$ which is an $M$-augmenting path of length at least $5$. Then, we have \(\cost (C) \geq 3 |\OPT \cap C|\). 
\end{lemma}

\begin{proof}
Let $C$ be of the form $a_0-b_0-a_1-\ldots-a_k-b_k$ be an $M$-augmenting path in $M+ \OPT$ of length $2k+1$, $k\geq 2$, where $a_0\in A$. Then, we have

\begin{align*}
\cost(C)=\underbrace{\cost(\{a_0,b_0\})}_{\substack{\ge 2\\\text{by Lemma~\ref{lem:start}}}}+\sum_{i=1}^{k-1} \underbrace{\cost(\{a_i,b_i\})}_{\substack{\ge 3\\\text{by Lemma~\ref{lem:3cost}}}}+ \underbrace{\cost(\{a_k,b_k\})}_{\substack{\ge 3\\\text{by Remark~\ref{rem:lastfour}}}}\geq
2+3(k-2)+3\\=3k-1=3|\OPT|-1\,.
\end{align*}

By Lemma~\ref{lem:start} and Corollary~\ref{cor:induction}, the above inequality is tight only if for each $i=0,\ldots,k-2$ we have that \(b_i\) rejected a proposal from \(a_i\) during the algorithm and $b_i$ points right.  However, in that case, we have $\cost(\{a_{k-1}, b_{k-1}, a_k, a_k\})\geq 7$ by Lemma~\ref{lem:end}, implying the desired inequality
\begin{align*}
\cost(C)=\underbrace{\cost(\{a_0,b_0\})}_{\substack{\ge 2\\\text{by Lemma~\ref{lem:start}}}}+\sum_{i=1}^{k-2} \underbrace{\cost(\{a_i,b_i\})}_{\substack{\ge 3\\\text{by Lemma~\ref{lem:3cost}}}}+ \underbrace{\cost(\{a_{k-1}, b_{k-1}, a_k, a_k\})}_{\ge7}\geq\\
2+3(k-3)+7=3k=3|\OPT|\,
\end{align*}
and finishing the proof.

\end{proof}

Our main theorem directly follows from Remark~\ref{rem:trivial_component}, Corollary~\ref{cor:nonaug}, Lemma~\ref{lem:low_cost}, Lemma~\ref{lemma:three_path_original_paper} and Lemma~\ref{lem:aug}.

\begin{theorem}\label{thm:main}
\(|\OPT| / |M| \leq 4/3\).
\end{theorem}

The next example demonstrates that the bound in Theorem~\ref{thm:main} is tight. In the description of this example we closely follow the description style from~\cite{Bauckholt}.

\begin{example}
In Figure~\ref{figure: tight_instance}, the circle nodes represent men and the square nodes represent women. The solid lines represent the edges in $G'$. The arrow tips indicate the preferences of each person. For example, $b_0$ has double-tipped arrows pointing to $a_1$ and $a_3$ and a single-tipped arrow pointing to $a_0$, so $b_0$ is indifferent between $a_1$ and $a_3$ and prefers both to $a_0$.

It is straightforward to verify that there exists a unique maximum cardinality stable matching, namely $\OPT=\{(a_0,b_0), (a_1,b_1), (a_2,b_2), (a_3,b_3)\}$. We can prove that there exists an execution of the algorithm with the above intermediate graph $G'$. Hence, for this execution one of the possible outputs is the matching $M=\{(a_1,b_0), (a_2, b_1),(a_3, b_3)\}$, leading to the ratio $|\OPT| / |M| = 4/3$.

\begin{figure}
\begin{center}
\begin{tikzpicture}[scale=1.7]
\tikzstyle{every node}=[rectangle, fill = black!30!white, draw = black, minimum size = 8pt, line width = 0.7pt]
\tikzstyle{edge}=[> = angle 90, shorten >= 2pt, shorten < = 2pt, line width = 0.7pt]
\tikzstyle{man}=[circle]

\node[label=below:$a_0$, man] (a0) at (0, 0) {};
\node[label=below:$a_1$, man] (a1) at (1, 0) {};
\node[label=below:$a_2$, man] (a2) at (2, 0) {};
\node[label=below:$a_3$,  man] (a3) at (3, 0) {};

\node[label=$b_0$] (b0) at (0, 1) {};
\node[label=$b_1$] (b1) at (1, 1) {};
\node[label=$b_2$] (b2) at (2, 1) {};
\node[label=$b_3$] (b3) at (3, 1) {};

\draw[edge,    dashed,    >-<] (b0) -- (a0);
\draw[edge,        >>-<] (b0) -- (a1);
\draw[edge,        >-<] (b1) -- (a1);
\draw[edge,        >-<<] (b1) -- (a2);
\draw[edge,   dashed,      >-<] (b2) -- (a2);
\draw[edge,        >-<<] (b3) -- (a2);
\draw[edge,        >-<] (b3) -- (a3);
\draw [edge,       >>-<] (b0) to [out=225, in=90] (-.5,0) to [out=270, in=225]  (a3);
\end{tikzpicture}
\end{center}
\caption{An instance for which the algorithm in~\cite{HuangKavitha2015} outputs a stable matching $M$ with $|\OPT| / |M| = 4/3$. }
\label{figure: tight_instance}
\end{figure}
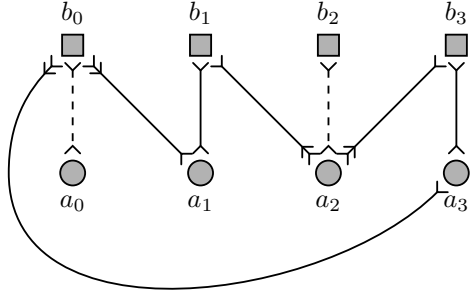
\end{example}
\begin{proof}
Let us provide an execution of the algorithm leading to the above intermediate graph $G'$.
\begin{itemize}
\item $a_3$ proposes to $b_3$; $b_3$ accepts.
\item $a_3$ proposes to $b_0$; $b_0$ accepts.
\item $a_2$ proposes to $b_3$; $b_3$ accepts.
\item $a_2$ proposes to $b_1$; $b_1$ accepts.
\item $a_1$ proposes to $b_1$; $b_1$ accepts.
\item $a_1$ proposes to $b_0$; $b_0$ accepts.
\item $a_0$ starts to propose to $b_0$, but each time $a_0$ makes a proposal, the proposal is rejected. $a_0$ gives up.
\end{itemize}
\end{proof}

\section*{Acknowledgements}

We would like to thank Laura Sanit\`{a} for suggesting us the maximum cardinality stable matching problem with ties of size two.

\bibliographystyle{plain}

\newpage
\bibliography{bib}

\section*{Appendix}

\subsection*{Comparison with Related Works}

The original paper~\cite{HuangKavitha2015} by Huang and Kavitha contained two algorithms: one for the problem with one-sided ties and one for the problem with ties of size two. Both algorithms consist of a series of proposals; and an output matching is constructed from the proposals accepted at the end of the algorithm.   Recently, Bauckholt, Pashkovich and  Sanit\`{a}~\cite{Bauckholt}  provided a tight analysis of the algorithm from~\cite{HuangKavitha2015} for the case of one-sided ties. In our paper and in both~\cite{Bauckholt} and \cite{HuangKavitha2015},  the analysis is based on charging schemes: first giving original charge to some objects and then redistributing this charge to nodes. The charging scheme for the problem with one-sided ties from~\cite{HuangKavitha2015} was substantially modified in~\cite{Bauckholt}. We give a new charging scheme which leads to a tight analysis for the problem with ties of size two. 

Below we compare our charging scheme to those from the papers~\cite{Bauckholt} and~\cite{HuangKavitha2015}.

\subsubsection*{Charging Scheme: Origin of Charges}
In analyzing the approximation ratio, we first fix an optimal matching.
In both~\cite{Bauckholt} and~\cite{HuangKavitha2015}, the charges are created by $5$-augmenting paths for the output matching with respect to the fixed optimal matching. In~\cite{Bauckholt}, the original charge is given to one man (the man in a ``$y$-node") on each $5$-augmenting path (see~\cite{Bauckholt}, Section 4, page 11). In~\cite{HuangKavitha2015}, the original charge is also  given to  each $5$-augmenting path (see~\cite{HuangKavitha2015}, Section 2, page 366 and Section 3, page 377).  In our charging scheme, the original charges are given to proposals. In particular, two charges are given to each of the proposals accepted at the of the algorithm (one for the man  and one for the woman participating in this proposal), except some special cases (if a proposal corresponds to a ``bad output" it is given an extra charge for the man, if it corresponds to a ``good input" it is given no charge for the woman). Because, the total number of ``bad outputs" is at most the number of ``good inputs", the total number of charges in our approach is bounded from above by two times the number of proposals accepted at the end of the algorithm. \emph{Hence, our charging scheme unlike charging schemes in~\cite{Bauckholt} and~\cite{HuangKavitha2015} does not provide a mapping from $5$-augmenting paths.} As a result, the total number of charges generated in~\cite{Bauckholt} and~\cite{HuangKavitha2015} is equal to the number of $5$-augmenting paths, while in our paper we know only that the total number of charges is bounded from above by two times  the number of the proposals accepted at the end of the algorithm (that is why the total number of charges generated in our scheme is at most four times the size of the output matching).

\subsubsection*{Charging Scheme: Distribution of Charges}
In both~\cite{Bauckholt} and~\cite{HuangKavitha2015}, the original charges are distributed ``globally".   In~\cite{HuangKavitha2015},  the original charges are first distributed to so called ``good paths", and then are redistributed from ``good paths" to nodes (see~\cite{HuangKavitha2015}, Section 2, page 366 and Section 3, page 377). In~\cite{Bauckholt}, the original charges are also distributed globally. First, each original charge received by a man is distributed to some woman locally, but then the charges are redistributed globally using ``path jumps", ``matching jumps" and  ``matching jumps with exceptions" (see~\cite{Bauckholt}, Section 4, page 11). ``Good paths", ``path jumps", ``matching jumps" and  ``matching jumps with exceptions" transfer the charges globally; i.e. they  can  transfer charges  from one part of the graph generated by the fixed optimal matching and the proposals, which are accepted at the end of the algorithm, to another part of this graph.  In our charging scheme, the original charges received by proposals are distributed locally. In particular, in our charging scheme each original charge of a proposal is distributed only  to nodes participating in this proposal. 

\emph{Our origin of charges together  with our distribution of charges form a ``local" charging scheme. In other words, the charges are generated by such objects as proposal (edges in the graph generated by the accepted proposals) and are distributed to nodes participating in these proposals (incident nodes). This makes our charging scheme easy to work with.}

\subsubsection*{Charging Scheme: Charges at the Beginning and at the End }
In order to compare the size of the output matching and the fixed optimal matching, it is natural to analyze the connected components in the union of these matchings. 
The approaches in~\cite{Bauckholt} and~\cite{HuangKavitha2015} were based on the following logic: the total number of generated charges was equal to the number of $5$-augmenting paths and then the charges were distributed to other connected components, so that  the total charge of each connected component was not too large with respect to the component's size. In other words, the approaches of~\cite{Bauckholt} and~\cite{HuangKavitha2015} were based on the idea that large connected components can ``fix the damages" caused by $5$-augmenting paths. \emph{Our approach is based on a different and  a more direct idea.}  We show that every connected component receives a total charge of at least three times the number of the edges from the fixed optimal matching in this component. Together with the fact that the total number of charges generated in our algorithm is at most four times the size of the output matching, we immediately infer the desired approximation ratio. 

\subsubsection*{Central Notions}
 
In~\cite{Bauckholt}, the notion of ``popularity" was introduced for women. Roughly speaking, a woman is called popular with respect to some man, if this woman holds two proposals at the end of the algorithm and both these proposals are not worse than the proposals that this man could offer (see~\cite{Bauckholt}, Section 3, page 6). This notion was important for the tight analysis provided in~\cite{Bauckholt}.  We use a different notion of ``popularity" in the current paper. For us, a woman is popular if this woman rejected a proposal during the algorithm. We also introduce further notions, for example ``successful". We call a woman (a man) ``successful" if this woman (this man) has two proposals (accepted) at the end of the algorithm. For working with connected components, we use the notion of ``pointing" similar to the notion of ``pointing" in~\cite{Bauckholt} (see~\cite{Bauckholt}, Section 4, page 14).

\end{document}